\documentclass[sigconf,natbib=true,review=false,anonymous=false,screen=True]{acmart}

\acmSubmissionID{16}

\AtBeginDocument{%
  \providecommand\BibTeX{{%
    \normalfont B\kern-0.5em{\scshape i\kern-0.25em b}\kern-0.8em\TeX}}}

\usepackage{acronym}
\usepackage{amsmath} % for "\DeclareMathOperator" macro

\usepackage{amssymb} % for "\mathbb" macro
\usepackage{acronym}
\usepackage{enumerate}
\usepackage{amsthm}
\usepackage{bbm}

\usepackage[inline]{enumitem}
\PassOptionsToPackage{dvipsnames}{xcolor}

\usepackage{dsfont}
\usepackage{acronym}
\usepackage[noend]{algorithmic}
\usepackage{algorithm}
\usepackage{bbm}
\usepackage{beramono}
\usepackage{bm}
\usepackage{booktabs}
\usepackage[skip=0pt]{caption}
\usepackage{cleveref}
\usepackage[T1]{fontenc}
\usepackage{graphicx}
\usepackage{hyperref}
\usepackage{listings}
\usepackage{multirow}
\usepackage{natbib}
\usepackage{subcaption}
\usepackage{tabularx}
\usepackage[dvipsnames]{xcolor}
\usepackage{enumerate}
\usepackage[inline]{enumitem}
\usepackage{numprint}
\usepackage[normalem]{ulem}
\usepackage{amsthm}
\usepackage{tikz}
\usetikzlibrary{automata, positioning, arrows}
\usepackage{xspace}

\usepackage{mathrsfs}
\usepackage{mathtools}
\usepackage[symbol]{footmisc}

\acrodef{CF}{collaborative filtering}
\acrodef{LTR}{learning to rank}
\acrodef{NDCG}{normalized discounted cumulative gain}
\acrodef{DCG}{discounted cumulative gain}
\acrodef{VAE}{variational autoencoder}
\acrodef{VAE}{variational autoencoder}
\acrodef{ELBO}{evidence lower bound objective}
\acrodef{IPS}{inverse propensity scoring}
\acrodef{BPR}{bayesian personalized ranking}
\acrodef{MF}{matrix factorization}
\acrodef{MNAR}{missing-not-at-random}
\acrodef{ULTR}{unbiased learning-to-rank}
\acrodef{CLTR}{counterfactual learning to rank}
\acrodef{LOLN}{law of large numbers}
\acrodef{CRM}{counterfactual risk minimization}
\acrodef{IS}{importance sampling}
\acrodef{i.i.d}{independent and identically distributed}
\acrodef{CRM}{counterfactual risk minimization}
\acrodef{PL}{Plackett-Luce}
\acrodef{CTR}{click through rate}
\acrodef{SEA}{safe exploration algorithm}
\acrodef{GENSPEC}{generalization and specialization }
\acrodef{VCRM}{variational counterfactual risk minimization}
\acrodef{SGD}{stochastic gradient descent}
\acrodef{DR}{doubly robust}
\acrodef{DM}{direct method}
\acrodef{ERC}{exposure ratio clipping}
\acrodef{IR}{information retrieval}

\theoremstyle{definition}

\newcommand{\headernodot}[1]{\vspace{1mm}\noindent\textbf{#1}}
\newcommand{\header}[1]{\headernodot{#1.}}

\allowdisplaybreaks

\setcopyright{rightsretained}

\copyrightyear{2025}
\acmYear{2025}
\setcopyright{cc}
\setcctype{by}
\acmConference[ICTIR '25]{Proceedings of the 2025 International ACM SIGIR
Conference on Innovative Concepts and Theories in Information Retrieval
(ICTIR)}{July 18, 2025}{Padua, Italy}
\acmBooktitle{Proceedings of the 2025 International ACM SIGIR Conference onInnovative Concepts and Theories in Information Retrieval (ICTIR) (ICTIR '25), July
18, 2025, Padua, Italy}\acmDOI{10.1145/3731120.3744583}
\acmISBN{979-8-4007-1861-8/2025/07}



\author{Shashank Gupta}
\authornote{\hspace{-1.99em}*Work done during an internship at Meta.}
\orcid{0000-0003-1291-7951}
\affiliation{%
	\institution{University of Amsterdam}
	\city{Amsterdam}
	\country{The Netherlands}
}
\email{s.gupta2@uva.nl}

\author{Yiming Liao}
\orcid{0000-0002-8480-0114}
\affiliation{%
	\institution{Meta AI}
	\city{New York}
	\country{United States of America}
}
\email{yimingliao@meta.com}

\author{Maarten de Rijke}
\orcid{0000-0002-1086-0202}
\affiliation{%
	\institution{University of Amsterdam}
	\city{Amsterdam}
	\country{The Netherlands}
}
\email{derijke@uva.nl}

\begin{document}

\newcommand\yiming[1]{\textcolor{blue}{[Yiming]: #1}}

\title[Towards Two-Stage Counterfactual Learning to Rank]{Towards Two-Stage Counterfactual Learning to Rank}

\begin{abstract}
\Ac{CLTR} aims to learn a ranking policy from user interactions while correcting for the inherent biases in interaction data, such as position bias. Existing \ac{CLTR} methods assume a single ranking policy that selects top-$K$ ranking from the entire document candidate set. In real-world applications, the candidate document set is on the order of millions, making a single-stage ranking policy impractical. In order to scale to millions of documents, real-world ranking systems are designed in a two-stage fashion, with a candidate generator followed by a ranker. The existing \ac{CLTR} method for a two-stage offline ranking system only considers the top-1 ranking set-up and only focuses on training the candidate generator, with the ranker fixed. A \ac{CLTR} method for training both the ranker and candidate generator jointly is missing from the existing literature.

In this paper, we propose a two-stage \ac{CLTR} estimator that considers the interaction between the two stages and estimates the joint value of the two policies offline. In addition, we propose a novel joint optimization method to train the candidate and ranker policies, respectively. To the best of our knowledge, we are the first to propose a \ac{CLTR} estimator and learning method for two-stage ranking. Experimental results on a semi-synthetic benchmark demonstrate the effectiveness of the proposed joint \ac{CLTR} method over baselines. 
\end{abstract}

\maketitle

\acresetall

\section{Introduction}
% \begin{figure}[h]
% \includegraphics[width=\columnwidth]{./images/Two Stage Slate (3).pdf}
% \end{figure}
A \ac{LTR} method aims to optimize a ranking policy to maximize a given \ac{IR} metric~\cite{liu2009learning}. Traditionally, \ac{LTR} policies were trained using manually curated queries and document relevance judgments. It is well known now that manually curating relevance judgments is time consuming, not scalable, and does not always translate to user preferences~\cite{sanderson2010test,chapelle2011yahoo}. As an alternative, \ac{LTR} from user interactions is scalable because user interactions are cheaper to collect on scale, and user interactions are generally aligned with user interests~\cite{karatzoglou2013learning}. However, user interactions are noisy and biased indicators of true user preferences, subject to biases such as position bias, and trust bias in search engines~\cite{agarwal2019addressing,joachims2017unbiased}. 

\Ac{CLTR} aims to learn a ranking policy while correcting for biases in user interaction data~\cite{joachims2017unbiased,oosterhuis2020unbiased,gupta2023recent,gupta2024unbiased}. \Ac{IPS} is the most common choice of estimator for \ac{CLTR}~\cite{oosterhuis2020policy,gupta2023safe}. \ac{IPS} weights each user interaction on a document with the inverse of the document's exposure probability, also known as document exposure propensity. Thus, documents with historically lower exposure propensity receive a higher weight and vice versa. In expectation, this procedure optimizes for unbiased document relevance~\cite{joachims2016counterfactual,joachims2017unbiased,agarwal2019addressing}.

Traditionally, \ac{CLTR} methods assume a single re-ranking policy, i.e., for a given query, the re-ranking policy selects top-$K$ items from the \emph{candidate document set}, generated via a (lightweight) candidate generator. In a typical \ac{CLTR} pipeline, the candidate generator is fixed, and only the re-ranker is updated with the \ac{CLTR} method~\cite{gupta2024unbiased}. To the best of our knowledge, there is no \ac{CLTR} method that jointly updates both candidate generator and the re-ranker.
% (or it While it works for a small document collection (for example, in the order of thousands), it can be impractical for a real-world setup where the document collection can be in the order of millions. In such cases, scoring each document and generating the top-$K$ list in near real-time is infeasible. Practically, two-stage \ac{LTR} systems are used to solve the large document set problem in industry applications~\cite{liu2009learning}.

In the context of offline contextual bandit learning, \citet{ma2020off} introduced the first two-stage off-policy correction method. The proposed estimator considered the interaction between the candidate generator and the contextual bandit policy. Although effective, the method has the following limitations: \begin{enumerate*}[label=(\roman*)]
    \item it is only designed for the contextual bandit setup (top-1 ranking) and extending it to the \ac{CLTR} setup is non-trivial since existing \ac{CLTR} estimators do no consider the interaction between different policies in a cascading setup~\cite{agarwal2019general,joachims2017unbiased,oosterhuis2020policy}; and 
    \item it only trains the candidate generator using the two-stage estimator while keeping the second-stage bandit policy fixed, typically pre-trained. 
\end{enumerate*}

We argue that this strategy is not optimal. Initially trained using logged data generated from the production candidate generator, the re-ranker faces potential performance issues when a new candidate generator is deployed with the new candidate policy. This causes a distribution shift in the input data distribution for the re-ranker, potentially leading to an overall degradation in the re-ranker performance when deployed online. To mitigate this, we propose a joint optimization method to optimize the re-ranker and candidate generator via an alternating optimization fashion. Experimental results on a semi-synthetic benchmark demonstrate that the proposed joint optimization method outperforms the two-staged baseline method, where the re-ranker and candidate generator are trained independently. To the best of our knowledge, this work is the first to develop a \ac{CLTR} estimator and learning method for the two-stage LTR system. We hope our contributions will enable \ac{LTR} practitioners to apply \ac{CLTR} methods at a real-world scale. 
% All source code to reproduce the synthetic experiments will
% be open-sourced upon acceptance.

\section{Related Work}
\subsection{\Acl{CLTR}}
In the context of \ac{LTR}, \citet{joachims2017unbiased} introduced the first counterfactual learning method to correct for position bias in search engines~\cite{joachims2002optimizing}. They applied the \ac{IPS} weighting technique, common in the offline contextual bandit literature~\cite{saito2021counterfactual}. For each user interaction with a document displayed by the ranking system, a click is weighted with the inverse of the document examination probability, a.k.a. document propensity. With the \ac{IPS} weighting scheme, the effect of the position bias is removed in expectation. As an extension, \citet{oosterhuis2020policy} introduced a policy-aware \ac{CLTR} method for the top-$K$ ranking setup. In a top-$K$ ranking setup, any document ranked beyond the position $K$ gets zero exposure. In this work, we deal with the top-$K$ setting. For a more recent overview of \ac{CLTR}, we refer to \citep{gupta2023recent}.

% \vspace{-2.3mm}
\subsection{Offline policy learning for contextual bandits}
In the context of offline contextual bandits, \ac{IPS} is commonly used to correct for selection bias over actions introduced by the previously deployed behavior policy~\cite{swaminathan2015batch,swaminathan2017off,xie2018off,schnabel2016recommendations,joachims2018deep,wu2018variance,gupta2023deep,gupta2024optimal}. Briefly, for each (action, context) pair, the corresponding reward is weighted by the ratio of corresponding action probabilities under the new policy and the behavior policy. In expectation, this weighting scheme removes the effect of selection bias introduced by the behavior policy. 

\vspace{-2.mm}
\subsection{Two-stage offline policy learning}
Most work on offline policy learning for contextual bandits and \ac{LTR} consider a single policy, typically the re-ranking policy, that re-ranks all the candidate documents for a given query context, generated via a candidate generation policy~\cite{saito2021counterfactual,gupta2023recent}. Typically, the candidate generation policy is fixed during the off-policy update of the re-ranker, and the interaction between the two policies is generally ignored.

In the context of \ac{LTR} from manual relevance judgments, \citet{dang2013two} introduced the first joint two-stage ranking method, which considers the interaction between the re-ranker and candidate policy and jointly updates both. \citet{wang2011cascade} introduced an efficient boosting-based method for joint two-staged ranking.

In offline contextual bandits, \citet{ma2020off} proposed a two-stage off-policy learning objective that considers the interaction between the first and second-stage contextual bandit policies. They propose a policy learning method to train the candidate generator considering the response of the re-ranker on the candidate generated at the first stage. 

\vspace{-2mm}
\section{Background}
\subsection{\Acl{CLTR}}
For a given query $q$ ($q \sim P(Q)$), a ranked list $y=\{d_1,d_2,\ldots,d_k\}$ generated by a ranker $\pi$, we define the utility of the list $y$ as:
\begin{equation}
    U(y \mid q, \pi) = \sum_{d \in D} \alpha_{k(d)} P(R = 1 \mid q, d) = \sum_{d \in D} \alpha_{k(d)} R_{q,d},
    \label{list-utility}
\end{equation}
where $k(d)$ is the rank of the document $d$ in the list $y$, and $\alpha_{k(d)}$ is the weight for the document $d$ at the rank $k(d)$, and $P(R = 1 \mid q, d)=R_{q,d}$ is the probability of relevance for the query, document pair $(q,d)$. Given the ranking utility, the goal of a \ac{LTR} algorithm is to train a ranking policy $\pi$ to optimize the following overall utility function: 
\begin{equation}
\begin{split}
    U(\pi) & = \mathbb{E}_{q \sim P(Q)}  \mathbb{E}_{y \sim \pi( \cdot \mid q)} \left[ U(y \mid q, \pi) \right] \\
    & = \mathbb{E}_{q} \left[ \sum_{d \in D} \rho(d \mid q, \pi) R_{q,d}  \right],
\end{split}    
    \label{true-utility}
\end{equation}
where the choice of the document weight $\rho(d \mid q, \pi)=\mathbb{E}_{y \sim \pi} \!\left[ \alpha_{k(d)} \right]$ (Eq.~\ref{list-utility}) defines the IR metric that is being optimized as a result. For example, the choice of $\alpha_{k(d)} = (\log_{2}(\textrm{rank}(d \mid y) + 1))^{-1}$ means that we are optimizing for the NDCG metric~\cite{liu2009learning}. 

Traditional \ac{LTR} methods assume access to the \emph{true} manually graded document relevance scores: $R_{q,d}$ (Eq.~\ref{list-utility}). However, manually curating relevance judgments in practice is not scalable, and relevance judgments are generally not aligned with user interests~\citep{chapelle2011yahoo,qin2010letor}. As an alternative, \ac{CLTR} considers the user interactions directly. \ac{CLTR} assumes access to the click log dataset $\mathcal{D}$ from a logging policy ($\pi_0$), $ \mathcal{D} = \big\{q_i, y_i, c_i \big\}^N_{i=1},
    \label{logs}$
with user issued queries $q_i$, the ranked list from the ranking policy $y_i$, and, finally, user clicks ($c_i \in {0,1}$) on the list $y_i$ indicate whether the document was clicked or not. In this work, we assume that user clicks follow the examination hypothesis~\cite{chuklin-click-2015,joachims2017unbiased}:
\begin{equation}
    P(C=1 \mid d, q, k) = P(E=1 \mid k)  P(R=1 \mid d, q),
    \label{click-model}
\end{equation}
which states that for a given query $q$, the probability of clicking on a document $d$ depends on the probability of examination at a given rank $k$ of the document and its relevance.

Given the click model (Eq.~\ref{click-model}), to define the \ac{IPS}-based \ac{CLTR} objective, following \citet{oosterhuis2020policy}, we define the document propensity under a policy $\pi$ as: 
\begin{equation}
\rho(d \mid q, \pi) = \mathbb{E}_{y \sim \pi} \left[  P(E=1 \mid k(d)) \right].
\label{eq:propensity_def}
\end{equation}
A similar propensity definition was used in other \ac{CLTR} work~\cite{yadav2021policy,gupta2023safe,oosterhuis2021unifying}. The choice of rank-based examination probability as the propensity sets the overall utility function as the expected number of clicks (Eq.~\ref{true-utility}). Finally, the \ac{CLTR} utility function is defined as:
\begin{equation}
    \hat{U}(\pi) = \frac{1}{N} \sum_{i=1}^{N} \sum_{d \in D}  \frac{\rho(d)}{\rho_{0}(d)} c_i(d),
    \label{cltr-obj}
\end{equation}
with $\rho_{0}(d)$ as the propensity score (Eq.~\ref{eq:propensity_def}) for the logging policy $\pi_0$.

\subsection{Two-stage off-policy learning}

In the context of offline contextual bandits, \citet{ma2020off} proposed a two-stage off-policy learning method, where the objective function is defined with respect to the bandit feedback dataset:
\begin{equation}
    \mathcal{D} = \big\{s_i, a_i, r_i \big\}^N_{i=1},
    \label{logs}
\end{equation}
where $s_i$ is the context/state, $a_i \sim \pi(\cdot \mid s_i)$ is the action sampled from the policy $\pi$, and $r_i$ is the reward for the (state, action) pair. The policy $\pi$ is defined as the mixture of the candidate policy which scores a size $k$ list: $A^k$,  $p(A^k \mid s)$ and the re-ranker $q(a \mid A^k, s)$: $\pi(a \mid s) = \sum_{A^k}^{} p(A^k \mid s) q(a \mid A^k, s)$.
Given this, the two-stage off-policy objective is defined as:
\begin{equation}
    \hat{U}(\pi) = \frac{1}{N} \sum_{i=1}^{N}  \frac{\pi(a_i \mid s_i)}{\pi_{0}(a_i \mid s_i)} r_i.
    \label{opl-obj}
\end{equation}
In their proposed method, \emph{only} the candidate policy $p(A^k \mid s)$ is updated via a stochastic gradient method; the re-ranker $q(a \mid A^k, s)$ is kept fixed. 

\section{Method: Two-stage Counterfactual Learning to Rank}
This section introduces our main contribution: a novel two-stage \ac{CLTR} objective and a joint learning method. 

% \vspace{-2mm}
\subsection{Two-stage \ac{LTR}}
We first define the two-stage \ac{LTR} objective with relevance labels. Given a candidate generator policy $\pi_c$, the candidate list $y_c \sim \pi_c$ of length $K_2$, the re-ranker policy $\pi_r$, and the final ranked list displayed to the user $y_r \sim \pi_r$ of length $K$, similar to the utility function for a given list in a single-stage setup $U(y \mid q, \pi)$ (Eq.~\ref{list-utility}), the two-stage \ac{LTR} objective is given by:
\begin{align}
    U(\pi_{c}, \pi_{r}) &= \mathbb{E}_{q \sim P(Q)}  \mathbb{E}_{y_c \sim \pi_{c}( \cdot \mid q)} \mathbb{E}_{y_r \sim \pi_{r}( \cdot \mid y_c, q)} \left[ U(y_r \mid q, \pi) \right] \nonumber \\
    &= \mathbb{E}_{q \sim P(Q)}  \mathbb{E}_{y_c \sim \pi_{c}( \cdot \mid q)} \mathbb{E}_{y_r \sim \pi_{r}( \cdot \mid y_c, q)} \left[ \sum_{d \in D} \alpha_{k(d)} R_{q,d} \right] \nonumber \\ 
    &= \mathbb{E}_{q \sim P(Q)}  \left[ \sum_{d \in D} \rho_{c,r}(d) R_{q,d} \right] ,
    \label{2stage-true-utility}
\end{align}
where $\rho_{c,r}(d) = \mathbb{E}_{y_c \sim \pi_{c}( \cdot \mid q)} \mathbb{E}_{y_r \sim \pi_{r}( \cdot \mid y_c, q)} \left[ \alpha_{k(d)} \right]$ is the new document weight, equivalent to the single-stage case (Eq.~\ref{true-utility}). Similar quantities appear in the literature on top-n recommendation evaluation~\cite{jeunen2024normalised}.

% \vspace{-2mm}
\subsection{Method: A novel two-stage \ac{CLTR}}
Given the true joint utility function for the two-stage ranking with the true relevance (Eq.~\ref{2stage-true-utility}), we now define the \ac{CLTR} objective with the logged data $\mathcal{D}$ (Eq.~\ref{logs}):
\begin{equation}
    \hat{U}(\pi_r, \pi_c) = \frac{1}{N} \sum_{i=1}^{N} \sum_{d \in D}  \frac{\rho_{r,c}(d)}{\rho_{0}(d)} c_i(d).
    \label{cltr-obj-2s}
\end{equation}
Similarly to the single-stage \ac{CLTR} objective (Eq.~\ref{cltr-obj}), we set the weight $\alpha_{k(d)} = P(E=1 \mid k(d)) $. As a result, the overall objective optimizes for expected clicks on relevant documents~\cite{oosterhuis2022doubly}. To the best of our knowledge, we are the first to introduce a two-staged \ac{CLTR} objective for real-world ranking systems.  

\begin{theorem}
\label{2s-unbiased}
The counterfactual objective for the two stage \ac{CLTR} method (Eq.~\ref{cltr-obj-2s}) is unbiased in expectation, i.e.,
\begin{equation}
   \mathbb{E}_{q, y \sim \pi_0, c} \left[ \hat{U}(\pi_r, \pi_c) \right] = U(\pi_r, \pi_c).
    \label{cltr-obj-unbias}
\end{equation}
\end{theorem}
\begin{proof}
    The expected value of the two-stage objective is given by:
    \begin{align}
   \mathbb{E}_{q, y \sim \pi_0, c} \left[ \hat{U}(\pi_r, \pi_c) \right] &= \mathbb{E}_{q} \left[ \sum_{d \in D}  \frac{\rho_{r,c}(d)}{\rho_{0}(d)} \mathbb{E}_{y \sim \pi_0, c} \left[ c(d) \right]  \right] \nonumber \\
   &= \mathbb{E}_{q} \left[ \sum_{d \in D}  \rho_{r,c}(d) R_{q,d}  \right],
    \end{align}
    where we apply the click model (Eq.~\ref{click-model}) in the first step.
\end{proof}

\noindent%
Intuitively, to maximize the objective in case of a click on a given document $d$, the optimizer will push the expected document weight under both policies $\rho_{r,c}(d)$ to the maximum. To maximize the expected weight, the candidate generator will push the document higher in the candidate list, which gives the re-ranker a chance to finally display the document to the user at a higher rank. 

The candidate generator can be selected as a lightweight model for efficiency, generating the candidate list quickly. The re-ranker can be a complex model that incorporates additional features to improve the precision of the results. However, for simplicity, we use the same model for both candidate generator and re-ranker and leave the effects of different architecture choices for future work. 
% \vspace{-2.5mm}
\subsection{Joint optimization for two-stage CLTR}
\label{proposed-method}
For optimization, we use stochastic gradient descent. To estimate the gradient with respect to the policy $\pi$, in this work, we make use of the general log-derivate trick of the REINFORCE algorithm~\cite{williams1992simple}, with the gradient of the metric weight term following:
\begin{equation}
\nabla_{\!\pi} \rho(d) = \mathbb{E}_{y \sim \pi} \big[ P(E = 1 \mid k)  \nabla_{\!\pi} \log \pi(y \mid q) \big].
\end{equation}
\citet{ma2020off} proposed to optimize the candidate generator policy while keeping the re-ranking policy constant during the optimization (the re-ranker is pre-trained on the same logged data). For the \ac{CLTR} context, this translates into estimating the gradient:
\begin{equation}
   \nabla_{\!\pi_c} \hat{U}(\pi_r, \pi_c) = \frac{1}{N} \sum_{i=1}^{N} \sum_{d \in D}  \frac{\nabla_{\!\pi_c} \rho_{r,c}(d)}{\rho_{0}(d)} c_i(d).
    \label{cltr-grad-2s}
\end{equation}
While this objective considers the feedback from the re-ranker for training the candidate generator, this strategy is overall sub-optimal. The re-ranker policy is pre-trained on the log data generated by the logging policy $\pi_0$. This means that it is optimized with respect to candidates generated by the logging policy during training. At inference time, if the candidate generator is switched to a different one, its performance might suffer because of the distribution shift. 

For two-stage optimization, if the re-ranker is kept fixed and the candidate is updated during training, then at test time, the distribution of the candidates will change, and the re-ranker might underperform. To remedy this problem, we propose a joint optimization method, where both the candidate and re-ranker are updated during training. We follow an alternating optimization method, wherein for each minibatch, the candidate generator and the re-ranker are updated alternatingly. This ensures that the re-ranker considers the candidate generator's feedback during training. 

\begin{table*}[h]
\setlength{\tabcolsep}{0.02cm}
\centering
\caption{
NDCG@10 performance of different two-stage \ac{CLTR} methods with varying candidate list sizes ($K_2$) and varying amounts of click logs ($N$). The numbers reported are averages over 25 independent runs.
}
\label{tab:dcgresults}
\resizebox{\textwidth}{!}{
\begin{tabular}{ l ccc  ccc ccc}
 \toprule
&\multicolumn{3}{c}{$K_{2}=500$}
&\multicolumn{3}{c}{$K_{2}=1,000$}
&\multicolumn{3}{c}{$K_{2}=1,500$}\\
\cmidrule(r){2-4}
\cmidrule(r){5-7}
\cmidrule{8-10}
Method & $N=0.1M$&  $N=0.32M$ &  $N=1M$&  $N=0.1M$ &  $N=0.32M$ &  $N=1M$&  $N=0.1M$ &  $N=0.32M$ &  $N=1M$\\
 \midrule
 % \multicolumn{1}{l}{Prod.}&  0.325 \small (0.011) &  \textbf{0.342 \small (0.012)} &  \textbf{0.330 \small (0.007)} &  0.448 \small (0.025) &  \textbf{0.443 \small (0.016)} &  \textbf{0.481 \small (0.001)} &  0.576 \small (0.010) &  \textbf{0.696 \small (0.001)} &  \textbf{0.706 \small (0.001)}  \\  
 \multicolumn{1}{l}{Baseline}&  0.404 \small (0.000) &  0.437 \small (0.001) &  0.470 \small (0.001) &  0.401 \small (0.001) &  0.440 \small (0.001) &  0.469 \small (0.001) &  0.400 \small (0.001) &  0.441 \small (0.001) &  0.474 \small (0.002)  \\ 
 \multicolumn{1}{l}{Indepen.}&  0.410 \small (0.001) &  0.454 \small (0.001) &  0.489 \small (0.001) &  0.416 \small (0.001) &  0.465 \small (0.015) &  0.496 \small (0.001) &  0.416 \small (0.001) &  0.455 \small (0.001) &  0.496 \small (0.001)  \\ 
 \multicolumn{1}{l}{Joint opt.}&  \textbf{0.420 \small (0.001)} &  \textbf{0.470 \small (0.001)} &  \textbf{0.504 \small (0.001)} &  \textbf{0.426 \small (0.001)} &  \textbf{0.479 \small (0.001)} &  \textbf{0.504 \small (0.001)} &  \textbf{0.424 \small (0.001)} &  \textbf{0.472 \small (0.001)} &  \textbf{0.510 \small (0.001)}  \\ 
 \bottomrule
\end{tabular}
    }
\end{table*}

\header{Gradient calculation}
As the choice of the ranking policies in the candidate generator and re-ranker, we choose a Plackett-Luce model, similar to previous work in \ac{CLTR}~\cite{oosterhuis2021computationally,oosterhuis2022doubly,gupta2023safe}. For optimization, we use stochastic gradient descent. For the gradient of the document weight with respect to the candidate policy, we apply the log-derivative/REINFORCE trick~\cite{williams1992simple,yadav2021policy,gupta2023safe}. The gradient expression is given as:
\begin{equation}
\label{cand-grad}
\begin{split}
& \nabla_{\!\pi_c} \rho_{r,c}(d) = {}\\
& \mbox{}\hspace*{3mm}
\mathbb{E}_{y_c \sim \pi_c} \big[ \mathbb{E}_{y_r \sim \pi_r(\cdot \mid y_c, q)} \left[ P(E = 1 \mid k) \right] \nabla_{\!\pi_c} \log \pi_c(y_c \mid q) \big].
\end{split}
\end{equation}
Similarly, the gradient of the document weight with respect to the re-ranker policy can be expressed as:
\begin{equation}
\label{ranker-grad}
\begin{split}
\mbox{}\hspace*{-2mm}
& \nabla_{\!\pi_r} \rho_{r,c}(d) = {}\\
\mbox{}\hspace*{-2mm}
& \mbox{}\hspace*{2mm}
\mathbb{E}_{y_c \sim \pi_c} \big[ \mathbb{E}_{y_r \sim \pi_r(\cdot \mid y_c, q)} \left[ P(E = 1 \mid k)   \nabla_{\!\pi_r}  \log \pi_r(y_r \mid y_c, q) \right] \big].
\end{split}
\end{equation}

\section{Experimental Setup}
For our experiments, we follow the semi-synthetic experimental setup, common in the \ac{CLTR} literature~\cite{joachims2017unbiased,oosterhuis2020policy,gupta2023safe,gupta2024practical}. The standard \ac{LTR} datasets involve a query and a pre-selected document list with the corresponding relevance judgments, for example, the MSLR30K dataset~\cite{qin2010letor}. The logging policy is trained on 3\% of the queries to simulate a production ranker, which is used to generate the logged dataset following a position-based user  model~\cite{craswell2008experimental,oosterhuis2020policy,joachims2017unbiased}. In the \ac{LTR} datasets, documents are already prefiltered from the entire document pool and in the order of $\mathbf{O}(100)$, which is not reflective of a real-world setup, where the candidate document set is of much larger order. As a result, the \ac{LTR} datasets are unsuitable for a realistic large-scale simulation of a two-staged \ac{CLTR} system. Instead, we rely on the MovieLens-1M recommender systems dataset for the simulation in this work. The dataset consists of user ratings of items. Following previous work~\cite{zou2020neural,zhang2019deep,zhao2013interactive}, we assume that the ratings reflect the \emph{true} user intent and are not biased by the recommendations to the user. We assume a rating value of $>$ 3 as a positive reward and a zero otherwise. We randomly sample $10\%$ of users with the \emph{true} rating as the ground truth for policy evaluation. This simulates a real-world setting where a production policy is trained on a small fraction of the \emph{ideal} relevance labels, then deployed to collect click signals~\cite{joachims2017unbiased,agarwal2019general}. We simulate a top-$K$ ranking setup~\cite{oosterhuis2020policy,gupta2023safe} where any item beyond the top-$K$ slot gets no exposure. The clicks are simulated via a position-based click model~\cite{craswell2008experimental,joachims2017unbiased}, where we define the examination probability as follows: 
\begin{equation}
    P(E=1 \mid q, d, y) = 
\begin{cases}
    \textrm{rank}(d \vert y)^{-1}& \text{if } \textrm{rank}(d \mid y) \leq 10,\\
    0              & \text{otherwise}.
\end{cases}    
\end{equation}
For our experiments, we define the query $q$ as the user context, and an item is a document $d$. Following previous work~\cite{oosterhuis2021robust,oosterhuis2021unifying,oosterhuis2022doubly,gupta2023safe}, we use frequency estimates for the propensity $\rho_0(d)$ (Eq.~\ref{cltr-grad-2s}). We use a simple matrix factorization model for the re-ranker and candidate generator policy, initialized with user and item embeddings of dimension 50, generated via applying SVD on the initial user-item rating matrix. Exploring the effect of different model architectures is out of scope and will be in future work. For estimating the policy gradient expectations (Eq.~\ref{cand-grad}, \ref{ranker-grad}), we use 300 Monte-Carlo samples~\cite{yadav2021policy,oosterhuis2022doubly,gupta2023safe}. For optimization, we use the Adam optimizer~\cite{kingma2014adam} with a learning rate of $0.01$ for all methods. Finally, the following methods are included in our comparisons:
\begin{enumerate}[label=(\arabic*), leftmargin=*]
    \item \emph{Baseline}. As a baseline, we follow the setup from previous work~\cite{ma2020off}, wherein we pre-train a re-ranker on the click data and only train the candidate generator using the two-stage objective (Eq.~\ref{cltr-obj-2s}) and evaluate via the two-stage \ac{LTR} objective (Eq.~\ref{2stage-true-utility}).  
     \item  \emph{Independent.} We train the candidate generator and reranker independently and evaluate via the two-stage objective (Eq.~\ref{2stage-true-utility}).
    \item  \emph{Joint optimization.} Our proposed method (Sect.~\ref{proposed-method}) is where we jointly optimize the candidate generator and re-ranker in an alternating optimization fashion.
\end{enumerate}
%

% \vspace{-1.5mm}
\section{Results}
Table~\ref{tab:dcgresults} reports the NDCG@10 results for different methods. We report the results with varying the candidate list sizes $K_2 \in \{500$, $1000$, $1500\}$ for realistic real-world evaluation. The table shows that the joint optimization method (Sect.~\ref{proposed-method}) performs the best compared to the other methods. 

The baseline method from the two-stage bandit work~\cite{ma2020off}, where the re-ranker is pretrained from the data generated from the logging candidate policy, performs worst among all methods, confirming our initial hypothesis that changing the candidate ranking distribution from training to inference can result in an overall performance drop for the re-ranker policy. The proposed joint optimization method for the candidate generator and the re-ranker performs the best across different candidate list sizes and logged data. This shows that the joint optimization strategy is optimal for the two-stage \ac{CLTR} methods. The performance of each method improves with increasing logged data size ($N=0.1M, 0.32M, 1M$), consistent with previous work~\cite{oosterhuis2020policy,joachims2017unbiased,gupta2023safe,jagerman2020safe,oosterhuis2021robust}.

Note that in this work, we use the same architecture for the candidate generator and the re-ranker, i.e., a matrix factorization model. Exploring the effects of different neural architectures for the candidate generator and re-ranker is beyond the scope of this work, as our focus is on introducing joint optimization for the two-stage CLTR. We will study it in future work. 

% \vspace{-1.5mm}
\section{Conclusion}
In this work, we have studied two-stage ranking systems from a \ac{CLTR} perspective. This is an important problem that has not been adequately discussed in the contemporary literature of CLTR; in this work, we address this gap. First, we propose a novel single \ac{LTR} two-stage objective function (Eq.~\ref{2stage-true-utility}) for optimizing \ac{LTR} systems with access to the \emph{ideal} relevance judgments. Next, we propose a two-stage \ac{CLTR} objective (Eq.~\ref{cltr-obj-2s}) that accounts for both the candidate generator and re-ranker jointly. To optimize the overall system, we propose a joint optimization strategy. We also provide the gradients for first-order optimization methods.

To the best of our knowledge, we are the first to tackle the problem of two-stage \ac{LTR} systems and a counterfactual learning and evaluation method for two-staged systems. As part of future work, we wish to extend the formulation from a contextual bandit setup to a reinforcement learning setup~\cite{gao2023alleviating}.

\begin{acks}
    This research was supported by the Dutch Research Council (NWO), under project numbers 024.004.022, NWA.1389.20.\-183, and KICH3.\-LTP.\-20.\-006, and the European Union's Horizon Europe program under grant agreement No 101070212.
    All content represents the opinion of the authors, which is not necessarily shared or endorsed by their respective employers and/or sponsors.
\end{acks}

% \clearpage
\bibliographystyle{ACM-Reference-Format}
\balance
\bibliography{references}

@article{liu2009learning,
  title={Learning to Rank for Information Retrieval},
  author={Liu, Tie-Yan},
  journal={Foundations and Trends in Information Retrieval},
  volume={3},
  number={3},
  pages={225--331},
  year={2009},
  publisher={Now Publishers, Inc.}
}

@article{jagerman2020safe,
	author = {Jagerman, Rolf and Markov, Ilya and de Rijke, Maarten},
	journal = {ACM Transactions on Information Systems (TOIS)},
	number = {3},
	pages = {1--23},
	publisher = {ACM New York, NY, USA},
	title = {Safe Exploration for Optimizing Contextual Bandits},
	volume = {38},
	year = {2020}}

@article{oosterhuis2022doubly,
  title={Doubly Robust Estimation for Correcting Position Bias in Click Feedback for Unbiased Learning to Rank},
  author={Oosterhuis, Harrie},
  journal={ACM Transactions on Information Systems},
  volume={41},
  number={3},
  pages={1--33},
  year={2023},
  publisher={ACM New York, NY}
}

@inproceedings{craswell2008experimental,
	author = {Craswell, Nick and Zoeter, Onno and Taylor, Michael and Ramsey, Bill},
	booktitle = {Proceedings of the 2008 international conference on web search and data mining},
	pages = {87--94},
	title = {An Experimental Comparison of Click Position-bias Models},
	year = {2008}}

@article{swaminathan2015batch,
	author = {Swaminathan, Adith and Joachims, Thorsten},
	journal = {The Journal of Machine Learning Research},
	number = {1},
	pages = {1731--1755},
	publisher = {JMLR. org},
	title = {Batch Learning from Logged Bandit Feedback through Counterfactual Risk Minimization},
	volume = {16},
	year = {2015}}

@article{zhang2019deep,
	author = {Zhang, Shuai and Yao, Lina and Sun, Aixin and Tay, Yi},
	journal = {ACM Computing Surveys (CSUR)},
	number = {1},
	pages = {1--38},
	publisher = {ACM New York, NY, USA},
	title = {Deep Learning Based Recommender System: A Survey and New Perspectives},
	volume = {52},
	year = {2019}}

@inproceedings{wu2018variance,
	author = {Wu, Hang and Wang, May},
	booktitle = {International Conference on Machine Learning},
	organization = {PMLR},
	pages = {5353--5362},
	title = {Variance Regularized Counterfactual Risk Minimization via Variational Divergence Minimization},
	year = {2018}}

@inproceedings{joachims2018deep,
	author = {Joachims, Thorsten and Swaminathan, Adith and de Rijke, Maarten},
	booktitle = {International Conference on Learning Representations},
	title = {Deep learning with Logged Bandit Feedback},
	year = {2018}}

@inproceedings{joachims2017unbiased,
	author = {Joachims, Thorsten and Swaminathan, Adith and Schnabel, Tobias},
	booktitle = {Proceedings of the Tenth ACM International Conference on Web Search and Data Mining},
	pages = {781--789},
	title = {Unbiased Learning-to-rank with Biased Feedback},
	year = {2017}}

@inproceedings{oosterhuis2020unbiased,
	author = {Oosterhuis, Harrie and Jagerman, Rolf and de Rijke, Maarten},
	booktitle = {Companion Proceedings of the Web Conference 2020},
	pages = {299--300},
	title = {Unbiased Learning to Rank: Counterfactual and Online Approaches},
	year = {2020}}

@inproceedings{agarwal2019addressing,
	author = {Agarwal, Aman and Wang, Xuanhui and Li, Cheng and Bendersky, Michael and Najork, Marc},
	booktitle = {The World Wide Web Conference},
	pages = {4--14},
	title = {Addressing Trust Bias for Unbiased Learning-to-rank},
	year = {2019}}

@inproceedings{schnabel2016recommendations,
	author = {Schnabel, Tobias and Swaminathan, Adith and Singh, Ashudeep and Chandak, Navin and Joachims, Thorsten},
	booktitle = {international conference on machine learning},
	organization = {PMLR},
	pages = {1670--1679},
	title = {Recommendations as Treatments: Debiasing Learning and Evaluation},
	year = {2016}}

@inproceedings{oosterhuis2021unifying,
	author = {Oosterhuis, Harrie and de Rijke, Maarten},
	booktitle = {Proceedings of the 14th ACM International Conference on Web Search and Data Mining},
	pages = {463--471},
	title = {Unifying Online and Counterfactual Learning to Rank: A Novel Counterfactual Estimator that Effectively Utilizes Online Interventions},
	year = {2021}}

@inproceedings{oosterhuis2020policy,
	author = {Oosterhuis, Harrie and de Rijke, Maarten},
	booktitle = {Proceedings of the 43rd International ACM SIGIR Conference on Research and Development in Information Retrieval},
	pages = {489--498},
	title = {Policy-aware Unbiased Learning to Rank for Top-k Rankings},
	year = {2020}}

@inproceedings{xie2018off,
	author = {Xie, Yuan and Liu, Boyi and Liu, Qiang and Wang, Zhaoran and Zhou, Yuan and Peng, Jian},
	booktitle = {International Conference on Learning Representations},
	title = {Off-Policy Evaluation and Learning from Logged Bandit Feedback: Error Reduction via Surrogate Policy},
	year = {2018}}

@inproceedings{joachims2002optimizing,
	author = {Joachims, Thorsten},
	booktitle = {Proceedings of the eighth ACM SIGKDD international conference on Knowledge discovery and data mining},
	pages = {133--142},
	title = {Optimizing Search Engines Using Clickthrough Data},
	year = {2002}}

@article{swaminathan2017off,
	author = {Swaminathan, Adith and Krishnamurthy, Akshay and Agarwal, Alekh and Dudik, Miro and Langford, John and Jose, Damien and Zitouni, Imed},
	journal = {Advances in Neural Information Processing Systems},
	title = {Off-Policy Evaluation for Slate Recommendation},
	volume = {30},
	year = {2017}}

@inproceedings{chapelle2011yahoo,
	author = {Chapelle, Olivier and Chang, Yi},
	booktitle = {Proceedings of the learning to rank challenge},
	organization = {PMLR},
	pages = {1--24},
	title = {Yahoo! Learning to Rank Challenge Overview},
	year = {2011}}

@article{qin2010letor,
	author = {Qin, Tao and Liu, Tie-Yan and Xu, Jun and Li, Hang},
	journal = {Information Retrieval},
	number = {4},
	pages = {346--374},
	publisher = {Springer},
	title = {LETOR: A Benchmark Collection for Research on Learning to Rank for Information Retrieval},
	volume = {13},
	year = {2010}}

@inproceedings{yadav2021policy,
	author = {Yadav, Himank and Du, Zhengxiao and Joachims, Thorsten},
	booktitle = {Proceedings of the 44th International ACM SIGIR Conference on Research and Development in Information Retrieval},
	pages = {1044--1053},
	title = {Policy-Gradient Training of Fair and Unbiased Ranking Functions},
	year = {2021}}

@inproceedings{joachims2016counterfactual,
	author = {Joachims, Thorsten and Swaminathan, Adith},
	booktitle = {Proceedings of the 39th International ACM SIGIR conference on Research and Development in Information Retrieval},
	pages = {1199--1201},
	title = {Counterfactual Evaluation and Learning for Search, Recommendation and Ad Placement},
	year = {2016}}

@inproceedings{saito2021counterfactual,
	author = {Saito, Yuta and Joachims, Thorsten},
	booktitle = {Fifteenth ACM Conference on Recommender Systems},
	pages = {828--830},
	title = {Counterfactual Learning and Evaluation for Recommender Systems: Foundations, Implementations, and Recent Advances},
	year = {2021}}

@book{chuklin-click-2015,
	author = {Chuklin, Aleksandr and Markov, Ilya and de Rijke, Maarten},
	publisher = {Morgan \& Claypool Publishers},
	series = {Synthesis Lectures on Information Concepts, Retrieval, and Services},
	title = {Click Models for Web Search},
	year = {2015}}

@inproceedings{oosterhuis2021computationally,
	author = {Oosterhuis, Harrie},
	booktitle = {Proceedings of the 44th International ACM SIGIR Conference on Research and Development in Information Retrieval},
	pages = {1023--1032},
	title = {Computationally Efficient Optimization of Plackett-Luce Ranking Models for Relevance and Fairness},
	year = {2021}}

@inproceedings{oosterhuis2021robust,
	author = {Oosterhuis, Harrie and de Rijke, Maarten de},
	booktitle = {Proceedings of the Web Conference 2021},
	pages = {158--170},
	title = {Robust Generalization and Safe Query-Specialization in Counterfactual Learning to Rank},
	year = {2021}}

@article{williams1992simple,
  title={Simple Statistical Gradient-Following Algorithms for Connectionist Reinforcement Learning},
  author={Williams, Ronald J},
  journal={Machine learning},
  volume={8},
  number={3},
  pages={229--256},
  year={1992},
  publisher={Springer}
}

@article{kingma2014adam,
  title={Adam: A Method for Stochastic Optimization},
  author={Kingma, Diederik P and Ba, Jimmy},
  journal={arXiv preprint arXiv:1412.6980},
  year={2014}
}

@article{sanderson2010test,
  title={Test Collection based Evaluation of Information Retrieval Systems},
  author={Sanderson, Mark},
  journal={Foundations and Trends in Information Retrieval},
  volume={4},
  number={4},
  pages={247--375},
  year={2010},
  publisher={Now Publishers, Inc.}
}

@inproceedings{gupta2023recent,
  title={Recent Advances in the Foundations and Applications of Unbiased Learning to Rank},
  author={Gupta, Shashank and Hager, Philipp and Huang, Jin and Vardasbi, Ali and Oosterhuis, Harrie},
  booktitle={Proceedings of the 46th International ACM SIGIR Conference on Research and Development in Information Retrieval},
  year={2023}
}

@inproceedings{gupta2023safe,
  title={Safe Deployment for Counterfactual Learning to Rank with Exposure-Based Risk Minimization},
  author={Gupta, Shashank and Oosterhuis, Harrie and de Rijke, Maarten},
  booktitle={Proceedings of the 46th International ACM SIGIR Conference on Research and Development in Information Retrieval},
  year={2023}
}

@inproceedings{agarwal2019general,
  title={A General Framework for Counterfactual Learning-to-rank},
  author={Agarwal, Aman and Takatsu, Kenta and Zaitsev, Ivan and Joachims, Thorsten},
  booktitle={Proceedings of the 42nd International ACM SIGIR Conference on Research and Development in Information Retrieval},
  pages={5--14},
  year={2019}
}

@inproceedings{karatzoglou2013learning,
  title={Learning to Rank for Recommender Systems},
  author={Karatzoglou, Alexandros and Baltrunas, Linas and Shi, Yue},
  booktitle={Proceedings of the 7th ACM Conference on Recommender Systems},
  pages={493--494},
  year={2013}
}

@inproceedings{ma2020off,
  title={Off-policy Learning in Two-stage Recommender Systems},
  author={Ma, Jiaqi and Zhao, Zhe and Yi, Xinyang and Yang, Ji and Chen, Minmin and Tang, Jiaxi and Hong, Lichan and Chi, Ed H.},
  booktitle={Proceedings of The Web Conference 2020},
  pages={463--473},
  year={2020}
}

@inproceedings{dang2013two,
  title={Two-stage Learning to Rank for Information Retrieval},
  author={Dang, Van and Bendersky, Michael and Croft, W. Bruce},
  booktitle={Advances in Information Retrieval: 35th European Conference on IR Research, ECIR 2013, Moscow, Russia, March 24-27, 2013. Proceedings 35},
  pages={423--434},
  year={2013},
  organization={Springer}
}

@inproceedings{gupta2023deep,
  title={A Deep Generative Recommendation Method for Unbiased Learning from Implicit Feedback},
  author={Gupta, Shashank and Oosterhuis, Harrie and de Rijke, Maarten},
  booktitle={Proceedings of the 2023 ACM SIGIR International Conference on Theory of Information Retrieval},
  pages={87--93},
  year={2023}
}

@inproceedings{gupta2024optimal,
  title={Optimal Baseline Corrections for Off-Policy Contextual Bandits},
  author={Gupta, Shashank and Jeunen, Olivier and Oosterhuis, Harrie and de Rijke, Maarten},
  booktitle={Proceedings of the 18th ACM Conference on Recommender Systems},
  pages={722--732},
  year={2024}
}

@inproceedings{zou2020neural,
  title={Neural Interactive Collaborative Filtering},
  author={Zou, Lixin and Xia, Long and Gu, Yulong and Zhao, Xiangyu and Liu, Weidong and Huang, Jimmy Xiangji and Yin, Dawei},
  booktitle={Proceedings of the 43rd International ACM SIGIR Conference on Research and Development in Information Retrieval},
  pages={749--758},
  year={2020}
}

@inproceedings{zhao2013interactive,
  title={Interactive Collaborative Filtering},
  author={Zhao, Xiaoxue and Zhang, Weinan and Wang, Jun},
  booktitle={Proceedings of the 22nd ACM international conference on Information \& Knowledge Management},
  pages={1411--1420},
  year={2013}
}

@inproceedings{gao2023alleviating,
  title={Alleviating Matthew Effect of Offline Reinforcement Learning in Interactive Recommendation},
  author={Gao, Chongming and Huang, Kexin and Chen, Jiawei and Zhang, Yuan and Li, Biao and Jiang, Peng and Wang, Shiqi and Zhang, Zhong and He, Xiangnan},
  booktitle={Proceedings of the 46th International ACM SIGIR Conference on Research and Development in Information Retrieval},
  pages={238--248},
  year={2023}
}

@inproceedings{gupta2024practical,
  title={Practical and Robust Safety Guarantees for Advanced Counterfactual Learning to Rank},
  author={Gupta, Shashank and Oosterhuis, Harrie and de Rijke, Maarten},
  booktitle={Proceedings of the 33rd ACM International Conference on Information and Knowledge Management},
  pages={737--747},
  year={2024}
}

@inproceedings{gupta2024unbiased,
  title={Unbiased Learning to Rank: On Recent Advances and Practical Applications},
  author={Gupta, Shashank and Hager, Philipp and Huang, Jin and Vardasbi, Ali and Oosterhuis, Harrie},
  booktitle={Proceedings of the 17th ACM International Conference on Web Search and Data Mining},
  pages={1118--1121},
  year={2024}
}

@inproceedings{wang2011cascade,
  title={A Cascade Ranking Model for Efficient Ranked Retrieval},
  author={Wang, Lidan and Lin, Jimmy and Metzler, Donald},
  booktitle={Proceedings of the 34th international ACM SIGIR conference on Research and development in Information Retrieval},
  pages={105--114},
  year={2011}
}

@inproceedings{jeunen2024normalised,
  title={On (normalised) discounted cumulative gain as an off-policy evaluation metric for top-n recommendation},
  author={Jeunen, Olivier and Potapov, Ivan and Ustimenko, Aleksei},
  booktitle={Proceedings of the 30th ACM SIGKDD conference on knowledge discovery and data mining},
  pages={1222--1233},
  year={2024}
}

\end{document}